\documentclass{article}%
\usepackage{amssymb}
\usepackage{amsmath}
\usepackage{amsfonts}
\usepackage{graphicx}%
\newtheorem{theorem}{Theorem}[section]

\newtheorem{corollary}[theorem]{Corollary}

\newtheorem{problem}[theorem]{Problem}

\newenvironment{proof}[1][Proof]{\noindent\textbf{#1.} }{\ \rule{0.5em}{0.5em}}
\begin{document}

\title{Vertices Belonging to All Critical Independent Sets of a Graph}
\author{Vadim E. Levit\\Ariel University Center of Samaria, Israel\\levitv@ariel.ac.il
\and Eugen Mandrescu\\Holon Institute of Technology, Israel\\eugen\_m@hit.ac.il}
\date{}
\maketitle

\begin{abstract}
Let $G=\left(  V,E\right)  $ be a graph. A set $S\subseteq V$ is
\textit{independent} if no two vertices from $S$ are adjacent, and by
$\mathrm{Ind}(G)$ ($\Omega(G)$) we mean the set of all (maximum) independent
sets of $G$, while $\mathrm{core}(G)=\cap\{S:S\in\Omega(G)\}$,
\cite{LevMan2002a}. The \textit{neighborhood} of $A\subseteq V$ is
$N(A)=\{v\in V:N(v)\cap A\neq\emptyset\}$. The \textit{independence number
}$\alpha(G)$ is the cardinality of each $S\in\Omega\left(  G\right)  $, and
$\mu(G)$ is the size of a maximum matching of $G$.

The number $id_{c}(G)=\max\{\left\vert I\right\vert -\left\vert
N(I)\right\vert :I\in\mathrm{Ind}(G)\}$ is called the \textit{critical
independence difference} of $G$, and $A\in\mathrm{Ind}(G)$ is
\textit{critical} if $\left\vert A\right\vert -\left\vert N(A)\right\vert
=id_{c}(G)$, \cite{Zhang}. We define $\mathrm{\ker}(G)=\cap\left\{  S:S\text{
\textit{is a critical independent set}}\right\}  $.

In this paper we prove that if a graph $G$ is non-quasi-regularizable (i.e.,
there exists some $A\in\mathrm{Ind}(G)$, such that $\left\vert A\right\vert
>\left\vert N(A)\right\vert $), then:

\begin{itemize}
\item $\mathrm{\ker}(G)\subseteq\mathrm{core}(G)$

\item $\left\vert \mathrm{\ker}(G)\right\vert >id_{c}\left(  G\right)
\geq\alpha\left(  G\right)  -\mu\left(  G\right)  \geq1$.

\end{itemize}

\textbf{Keywords:} independent set, critical set, critical difference, maximum matching

\end{abstract}

\section{Introduction}

Throughout this paper $G=(V,E)$ is a simple (i.e., a finite, undirected,
loopless and without multiple edges) graph with vertex set $V=V(G)$ and edge
set $E=E(G)$. We consider only graphs without isolated vertices.

If $X\subseteq V$, then $G[X]$ is the subgraph of $G$ spanned by $X$. By $G-W$
we mean either the subgraph $G[V-W]$, if $W\subseteq V(G)$, or the partial
subgraph $H=(V,E-W)$ of $G$, for $W\subseteq E(G)$. In either case, we use
$G-w$, whenever $W$ $=\{w\}$. If $X,Y\subset V$ are non-empty and disjoint,
then we denote $\left(  X,Y\right)  =\left\{  xy:xy\in E,x\in X,y\in
Y\right\}  $.

The \textit{neighborhood} of a vertex $v\in V$ is the set $N(v)=\{w:w\in V$
\textit{and} $vw\in E\}$, while the \textit{closed neighborhood} of $v\in V$
is $N[v]=N(v)\cup\{v\}$; in order to avoid ambiguity, we use also $N_{G}(v)$
instead of $N(v)$. In particular, if $\left\vert N(v)\right\vert =1$, then $v$
is a \textit{pendant vertex} of $G$, and \textrm{pend}$(G)=\{v\in V(G):v$
\textit{is a pendant vertex in} $G\}$. The \textit{neighborhood} of
$A\subseteq V$ is\emph{ }denoted by $N(A)=N_{G}(A)=\{v\in V:N(v)\cap
A\neq\emptyset\}$, and $N[A]=N(A)\cup A$.

A set $S\subseteq V(G)$ is \textit{independent} if no two vertices from $S$
are adjacent, and by $\mathrm{Ind}(G)$ we mean the set of all the independent
sets of $G$. An independent set of maximum size will be referred to as a
\textit{maximum independent set} of $G$, and the \textit{independence number
}of $G$ is $\alpha(G)=\max\{\left\vert S\right\vert :S\in\mathrm{Ind}(G)\}$. A
graph $G$ is \textit{quasi-regularizable} if one can replace each edge of $G$
with a non-negative integer number of parallel copies, so as to obtain a
regular multigraph of degree $\neq0$, \cite{berge2}. For instance, $K_{4}-e$,
$e\in E\left(  K_{4}\right)  $,\ is quasi-regularizable, while $P_{3}$ is not
quasi-regularizable. It is clear that a quasi-regularizable graph can not have
isolated vertices.

\begin{theorem}
\label{th11}For a graph $G$ the following assertions are equivalent:

\emph{(i)} quasi-regularizable;

\emph{(ii) }\cite{berge2} $\left\vert S\right\vert \leq\left\vert
N(S)\right\vert $ holds for every $S$ $\in\mathrm{Ind}(G)$;

\emph{(iii)} \cite{Tutte} $G$ has a perfect $2$-matching, i.e., $G$ contains a
system of vertex-disjoint odd cycles and edges covering all its vertices.
\end{theorem}

Let $\Omega(G)=\{S:S$ \textit{is a maximum independent set of} $G\}$ and
$\xi(G)=\left\vert \mathrm{core}(G)\right\vert $, where $\mathrm{core}%
(G)=\cap$ $\{S:S\in\Omega(G)\}$, \cite{LevMan2002a}.

Similarly, let $\mathrm{corona}(G)=\cup$ $\{S:S\in\Omega(G)\}$, and
$\zeta(G)=|\mathrm{corona}(G)|$, \cite{BorosGolLev}.

A \textit{matching} is a set of non-incident edges of $G$; a matching of
maximum cardinality $\mu(G)$ is a \textit{maximum matching}, and a
\textit{perfect matching} is a matching covering all the vertices of $G$.

In the sequel we need the following characterization of a maximum independent
set of a graph, due to Berge.

\begin{theorem}
\label{th2}\cite{berge2} An independent set $S$ belongs to $\Omega(G)$ if and
only if every independent set $A$ of $G$, disjoint from $S$, can be matched
into $S$.
\end{theorem}

$G$ is called a \textit{K\"{o}nig-Egerv\'{a}ry graph }provided $\alpha
(G)+\mu(G)=\left\vert V(G)\right\vert $ \cite{dem,ster}. It is known that each
bipartite graph satisfies this property.

\begin{theorem}
\label{th5} \cite{LevMan2003} If $G$ is a K\"{o}nig-Egerv\'{a}ry graph, $M$ is
a maximum matching, then $M$ matches $V\left(  G\right)  -S$ into $S$, for
every $S\in\Omega\left(  G\right)  $, and $\mu\left(  G\right)  =\left\vert
V\left(  G\right)  -S\right\vert $.
\end{theorem}

In Boros \emph{et al.} \cite{BorosGolLev} it has been proved that if $G$ is
connected and $\alpha(G)>\mu(G)$, then $\xi(G)=\left\vert \mathrm{core}%
(G)\right\vert >\alpha(G)-\mu(G)$. This strengthened the following finding
stated in \cite{LevMan2002a}: if $\alpha(G)>(\left\vert V(G)\right\vert
+k-1)/2$, then $\xi(G)\geq k+1$; moreover, $\xi(G)\geq k+2$ is valid, whenever
$\left\vert V(G)\right\vert +k-1$ is an even number. For $k=1$, the previous
inequality provides us with a generalization of a result of Hammer \emph{et
al.} \cite{HamHanSim} claiming that if a graph $G$ has $\alpha(G)>\left\vert
V\left(  G\right)  \right\vert /2$, then $\xi(G)\geq1$. In \cite{LevMan2001}
it was shown that if $G$ is a connected bipartite graph with $\left\vert
V(G)\right\vert \geq2$, then $\xi(G)\neq1$. Jamison \cite{Jamison}, Zito
\cite{Zito}, and Gunther \emph{et al.} \cite{GunHarRall} proved independently
that $\xi(G)\neq1$ is true for any tree $T$.

In Chleb\'{\i}k \emph{et al.} \cite{ChCh2008} it has been found that if there
is some $S\in\mathrm{Ind}(G)$, such that $\left\vert S\right\vert >\left\vert
N(S)\right\vert $, then $\left\vert \mathrm{core}(G)\right\vert >\max
\{\left\vert I\right\vert -\left\vert N(I)\right\vert :I\in\mathrm{Ind}(G)\}$.
It strengthens the inequality $\left\vert \mathrm{core}(G)\right\vert
>\alpha(G)-\mu(G)$ \cite{BorosGolLev}, since $\max\{\left\vert I\right\vert
-\left\vert N(I)\right\vert :I\in\mathrm{Ind}(G)\}\geq\alpha(G)-\mu(G)$
\cite{Lorentzen1966,Schrijver2003}.

The number $d\left(  X\right)  =\left\vert X\right\vert -\left\vert N\left(
X\right)  \right\vert $ is called the difference of the set $X\subseteq
V\left(  G\right)  $, and $d_{c}(G)=\max\{d(X):X\subseteq V\left(  G\right)
\}$ is the \textit{critical difference} of $G$. A set $U\subseteq V(G)$ is
\textit{critical} if $d(U)=d_{c}(G)$ \cite{Zhang}. The number $id_{c}%
(G)=\max\{d(I):I\in\mathrm{Ind}(G)\}$ is called the \textit{critical
independence difference} of $G$. If $A\subseteq V(G)$ is independent and
$d(A)=id_{c}(G)$, then $A$ is called \textit{critical independent
}\cite{Zhang}.

For a graph $G$ let us denote $\mathrm{\ker}(G)=\cap\left\{  S:S\text{
\textit{is a critical independent set}}\right\}  $ and $\varepsilon
(G)=\left\vert \mathrm{\ker}(G)\right\vert $.

For instance, the graph $G_{1}$ in Figure \ref{fig3} has $\mathrm{\ker}\left(
G_{1}\right)  =\mathrm{core}(G_{1})=\{a,b\}$. The graph $G_{2}$ from Figure
\ref{fig3} has $X=\left\{  x,y,z,p,q\right\}  $ as a critical non-independent
set, because $d(X)=1=d_{c}(G_{2})$, while $\mathrm{\ker}\left(  G_{2}\right)
=\{x,y\}\subset\mathrm{core}(G_{2})=\{x,y,z\}$. The graph $G_{3}$ from Figure
\ref{fig3} has $\{t,u,v\}$\ as a critical set, $\mathrm{\ker}\left(
G_{3}\right)  =\{u,v\}$, while $\mathrm{core}(G_{3})=\left\{  t,u,v,w\right\}
$\ is not a critical set.

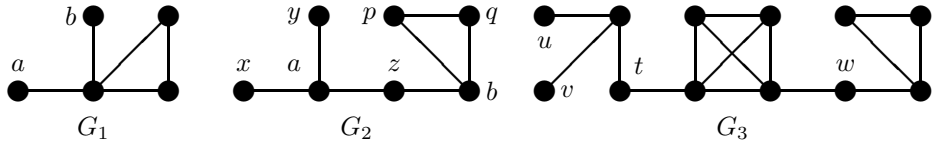
\begin{figure}[h]
\setlength{\unitlength}{1.0cm} \begin{picture}(5,1.9)\thicklines
\multiput(1,0.5)(1,0){3}{\circle*{0.29}}
\multiput(2,1.5)(1,0){2}{\circle*{0.29}}
\put(1,0.5){\line(1,0){2}}
\put(2,0.5){\line(0,1){1}}
\put(2,0.5){\line(1,1){1}}
\put(3,0.5){\line(0,1){1}}
\put(1,0.84){\makebox(0,0){$a$}}
\put(1.7,1.5){\makebox(0,0){$b$}}
\put(2,0){\makebox(0,0){$G_{1}$}}
\multiput(4,0.5)(1,0){4}{\circle*{0.29}}
\put(4,0.5){\line(1,0){3}}
\multiput(5,1.5)(1,0){3}{\circle*{0.29}}
\put(5,0.5){\line(0,1){1}}
\put(7,0.5){\line(0,1){1}}
\put(7,0.5){\line(-1,1){1}}
\put(6,1.5){\line(1,0){1}}
\put(6,0.84){\makebox(0,0){$z$}}
\put(4,0.84){\makebox(0,0){$x$}}
\put(4.67,1.5){\makebox(0,0){$y$}}
\put(4.67,0.84){\makebox(0,0){$a$}}
\put(5.67,1.5){\makebox(0,0){$p$}}
\put(7.3,1.5){\makebox(0,0){$q$}}
\put(7.3,0.5){\makebox(0,0){$b$}}
\put(5.5,0){\makebox(0,0){$G_{2}$}}
\multiput(8,0.5)(1,0){6}{\circle*{0.29}}
\multiput(8,1.5)(1,0){6}{\circle*{0.29}}
\put(8,0.5){\line(1,1){1}}
\put(8,1.5){\line(1,0){1}}
\put(9,0.5){\line(0,1){1}}
\put(9,0.5){\line(1,0){4}}
\put(10,0.5){\line(0,1){1}}
\put(10,0.5){\line(1,1){1}}
\put(10,1.5){\line(1,-1){1}}
\put(10,1.5){\line(1,0){1}}
\put(11,0.5){\line(0,1){1}}
\put(12,1.5){\line(1,-1){1}}
\put(12,1.5){\line(1,0){1}}
\put(13,0.5){\line(0,1){1}}
\put(8.3,0.5){\makebox(0,0){$v$}}
\put(8,1.15){\makebox(0,0){$u$}}
\put(9.25,0.83){\makebox(0,0){$t$}}
\put(12,0.84){\makebox(0,0){$w$}}
\put(10.5,0){\makebox(0,0){$G_{3}$}}
\end{picture}\caption{Non-quasi-regularizable graphs.}%
\label{fig3}%
\end{figure}

Clearly, $d_{c}(G)\geq id_{c}(G)$ is true for every graph $G$.

\begin{theorem}
\label{Theorem3}\cite{Zhang} The equality $d_{c}(G)$ $=id_{c}(G)$ holds for
every graph $G$.
\end{theorem}

If $A\in\Omega(G[N[A]])$, then $A$ is called a \textit{local maximum
independent set} of $G$ \cite{LevMan2002b}.

It is easy to see that all pendant vertices are included in every maximum
critical independent set. It is known that the problem of finding a critical
independent set is polynomially solvable \cite{Ageev,Zhang}.

\begin{theorem}
\label{th8}\emph{(i) }\cite{NemhTro} Each local maximum independent set is
included in a maximum independent set.

\emph{(ii)} \cite{LevManKE2009} Every critical independent set is a local
maximum independent set.

\emph{(iii)} \cite{Butenko} Each critical independent set is contained in some
maximum independent set.

\emph{(iv)} \cite{Larson} There is a matching from $N(S)$ into $S$, for every
critical independent set $S$.
\end{theorem}

In this paper we prove that $\mathrm{\ker}(G)\subseteq$ $\mathrm{core}(G)$ and
$\varepsilon\left(  G\right)  \geq d_{c}\left(  G\right)  \geq\alpha\left(
G\right)  -\mu\left(  G\right)  $ hold for every graph $G$.

\section{Results}

\begin{theorem}
\label{Theorem1}Let $A$ be a critical independent set of the graph $G$ and
$X=A\cup N\left(  A\right)  $. Then the following assertions are true:

\emph{(i)} $H=G\left[  X\right]  $ is a K\"{o}nig-Egerv\'{a}ry graph;

\emph{(ii)} $\alpha\left(  G\left[  V-X\right]  \right)  \leq\mu\left(
G\left[  V-X\right]  \right)  $;

\emph{(iii)} $\mu\left(  G\left[  X\right]  \right)  +\mu\left(  G\left[
V-X\right]  \right)  =\mu\left(  G\right)  $; in particular, each maximum
matching of $G\left[  X\right]  $ can be enlarged to a maximum matching of $G$.
\end{theorem}

\begin{proof}
\emph{(i)} By Theorem \ref{th8}\emph{(ii)}, $A$ is a local maximum independent
set, which ensures that $\alpha(H)=\left\vert A\right\vert $, while Theorem
\ref{th8}\emph{(iv)} implies $\mu(H)=\left\vert N(A)\right\vert $.
Consequently, we get that
\[
\alpha(H)+\mu(H)=\left\vert A\cup N(A)\right\vert =\left\vert X\right\vert
=\left\vert V(H)\right\vert ,
\]
i.e., $H$ is a K\"{o}nig-Egerv\'{a}ry graph.

\emph{(ii) }According to Theorem \ref{th8}\emph{(iii)}, there exists a maximum
independent set $S$ such that $A\subseteq S$. Suppose that $\left\vert
B\right\vert >\left\vert N\left(  B\right)  \right\vert $ holds for some
$B\subseteq S-A$. Then, it follows that
\[
\left\vert A\right\vert -\left\vert N\left(  A\right)  \right\vert <\left(
\left\vert A\right\vert -\left\vert N\left(  A\right)  \right\vert \right)
+\left(  \left\vert B\right\vert -\left\vert N\left(  B\right)  \right\vert
\right)  \leq\left\vert A\cup B\right\vert -\left\vert N\left(  A\cup
B\right)  \right\vert ,
\]
which contradicts the hypothesis on $A$, namely, the fact that $\left\vert
A\right\vert -\left\vert N\left(  A\right)  \right\vert =d_{c}(G)$. Hence
$\left\vert B\right\vert \leq\left\vert N\left(  B\right)  \right\vert $ is
true for every $B\subseteq S-A$. Consequently, by Hall's Theorem there exists
a matching from $S-A$ into $V-S-N\left(  A\right)  $ that implies $\left\vert
S-A\right\vert \leq\mu\left(  G\left[  V-X\right]  \right)  $.

It remains to show that $\alpha\left(  G\left[  V-X\right]  \right)
=\left\vert S-A\right\vert $. By way of contradiction, assume that
\[
\alpha\left(  G\left[  V-X\right]  \right)  =\left\vert D\right\vert
>\left\vert S-A\right\vert
\]
for some independent set $D\subseteq V-X$. Since $D\cap N\left[  A\right]
=\emptyset$, the set $A\cup D$ is independent, and
\[
\left\vert A\cup D\right\vert =\left\vert A\right\vert +\left\vert
D\right\vert >\left\vert A\right\vert +\left\vert S-A\right\vert
=\alpha\left(  G\right)  ,
\]
which is impossible.

\emph{(iii) }Let $M_{1}$ be a maximum matching of $H$ and $M_{2}$ be a maximum
matching of $G\left[  V-X\right]  $. We claim that $M_{1}\cup M_{2}$ is a
maximum matching of $G$. \begin{figure}[h]
\setlength{\unitlength}{1.0cm} \begin{picture}(5,4)\thicklines
\put(6.7,1){\oval(11.7,1.5)}
\put(5,1){\oval(4.5,1)}
\put(5,3){\oval(4.5,1)}
\multiput(3.5,1)(1,0){2}{\circle*{0.29}}
\multiput(3.5,3)(1,0){2}{\circle*{0.29}}
\multiput(6.5,1)(0,2){2}{\circle*{0.29}}
\multiput(4.6,1)(0.35,0){5}{\circle*{0.1}}
\multiput(4.6,3)(0.35,0){5}{\circle*{0.1}}
\put(6.5,1){\line(0,1){2}}
\multiput(3.5,1)(1,0){2}{\line(0,1){2}}
\put(0.5,1){\makebox(0,0){$S$}}
\put(0.2,2.2){\makebox(0,0){$G$}}
\put(1.75,1){\makebox(0,0){$S-A$}}
\put(1.5,3){\makebox(0,0){$V-S-N(A)$}}
\put(10,1){\oval(3.3,1)}
\put(9.5,3){\oval(2.3,1)}
\multiput(9,1)(1,0){3}{\circle*{0.29}}
\multiput(9,3)(1,0){2}{\circle*{0.29}}
\put(9,1){\line(0,1){2}}
\put(10,1){\line(0,1){2}}
\put(12,1){\makebox(0,0){$A$}}
\put(11.2,3){\makebox(0,0){$N(A)$}}
\put(8.2,2.2){\makebox(0,0){$M_{1}$}}
\put(2.5,2.2){\makebox(0,0){$M_{2}$}}
\end{picture}\caption{$S\in\Omega(G)$ and $A$ is a critical independent set of
$G${.}}%
\label{fig44}%
\end{figure}
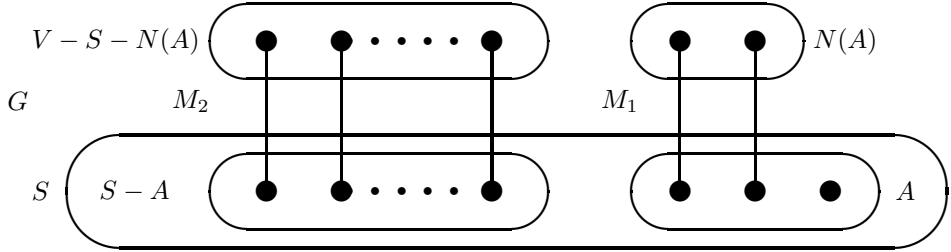

The only edges that may enlarge $M_{1}\cup M_{2}$ belong to the set $\left(
N\left(  A\right)  ,V-S-N\left(  A\right)  \right)  $. The matching $M_{1}$
covers all the vertices of $N\left(  A\right)  $ in accordance with Theorem
\ref{th5} and part \emph{(i)}. Therefore, to choose an edge from the set
$\left(  N\left(  A\right)  ,V-S-N\left(  A\right)  \right)  $ means to loose
an edge from $M_{1}$. In other words, no matching different from $M_{1}\cup
M_{2}$ may overstep $\left\vert M_{1}\cup M_{2}\right\vert $.

Consequently, each maximum matching of $G\left[  X\right]  $ can find its
counterpart in $G\left[  V-X\right]  $ in order to build a maximum matching of
$G$.
\end{proof}

Theorem \ref{Theorem1} allows us to give an alternative proof of the following
inequality due to Lorentzen.

\begin{corollary}
\label{cor1}\cite{Lorentzen1966}, \cite{Schrijver2003} The inequality
$d_{c}\left(  G\right)  \geq\alpha\left(  G\right)  -\mu\left(  G\right)  $
holds for every graph $G$.
\end{corollary}

\begin{proof}
Let $A$ be a critical independent set of $G$, and $X=A\cup N\left(  A\right)
$.

By Theorem \ref{Theorem1}\emph{(ii)}, we get $\alpha\left(  G\left[
V-X\right]  \right)  -\mu\left(  G\left[  V-X\right]  \right)  \leq0$. Hence
it follows that%
\[
\alpha\left(  G\left[  X\right]  \right)  -\mu\left(  G\left[  X\right]
\right)  \geq\left(  \alpha\left(  G\left[  X\right]  \right)  +\alpha\left(
G\left[  V-X\right]  \right)  \right)  -\left(  \mu\left(  G\left[  X\right]
\right)  +\mu\left(  G\left[  V-X\right]  \right)  \right)  .
\]

Theorem \ref{Theorem1}\emph{(iii) }claims that $\mu\left(  G\left[  X\right]
\right)  +\mu\left(  G\left[  V-X\right]  \right)  =\mu\left(  G\right)  $.

Since $A$ is a critical independent set, there exists some $S\in\Omega\left(
G\right)  $ such that $A\subseteq S$, and $\alpha\left(  G\left[  X\right]
\right)  =\left\vert A\right\vert $, by Theorem \ref{th8}\emph{(i)}. Hence we
have%
\[
\alpha\left(  G\left[  X\right]  \right)  +\alpha\left(  G\left[  V-X\right]
\right)  =\left\vert A\right\vert +\left\vert S-A\right\vert =\alpha\left(
G\right)  .
\]
In addition,\emph{ }Theorem \ref{Theorem1}\emph{(i)} and Theorem
\ref{th5}\emph{ }imply that $\mu\left(  G\left[  X\right]  \right)
=\left\vert N(A)\right\vert $.

Finally, we obtain
\begin{align*}
d_{c}\left(  G\right)   &  =\max\left\{  \left\vert I\right\vert -\left\vert
N(I)\right\vert :I\in\mathrm{Ind}(G)\right\}  =\left\vert A\right\vert
-\left\vert N(A)\right\vert =\\
&  =\alpha\left(  G\left[  X\right]  \right)  -\mu\left(  G\left[  X\right]
\right)  \geq\alpha\left(  G\right)  -\mu\left(  G\right)  ,
\end{align*}
and this completes the proof.
\end{proof}

Applying Theorem \ref{Theorem1} and Theorem \ref{th8}\emph{(iii) }we get the following.

\begin{corollary}
\cite{Larson2011} Let $J$ be a maximum critical independent set of $G$, and
$X=J\cup N(J)$. Then the following assertions are true:

\emph{(i)} $\alpha(G)=\alpha(G[X])+\alpha(G[V-X])$;

\emph{(ii)} $\alpha(G)=\alpha_{c}(G)+\alpha(G[V-X])$;

\emph{(iii)} $G[X]$ is a K\"{o}nig-Egerv\'{a}ry graph.
\end{corollary}

The graph $G$ from Figure \ref{fig24} has $\mathrm{\ker}(G)=\left\{
a,b,c\right\}  $. Notice that $\mathrm{\ker}(G)\subseteq\mathrm{core}(G)$;
$S=\left\{  a,b,c,v\right\}  $ is a largest critical independent set, and
neither $S\subseteq\mathrm{core}(G)$ nor $\mathrm{core}(G)\subseteq S$. In
addition, $\mathrm{core}(G)$ is not a critical independent set of $G$.
\begin{figure}[h]
\setlength{\unitlength}{1cm}\begin{picture}(5,1.2)\thicklines
\multiput(3,0)(1,0){9}{\circle*{0.29}}
\multiput(3,1)(1,0){4}{\circle*{0.29}}
\multiput(9,1)(2,0){2}{\circle*{0.29}}
\put(3,0){\line(1,0){8}}
\put(4,0){\line(0,1){1}}
\put(3,1){\line(1,-1){1}}
\put(5,0){\line(0,1){1}}
\put(5,1){\line(1,0){1}}
\put(6,1){\line(1,-1){1}}
\put(9,0){\line(0,1){1}}
\put(9,1){\line(1,-1){1}}
\put(11,0){\line(0,1){1}}
\put(2.7,0){\makebox(0,0){$a$}}
\put(2.7,1){\makebox(0,0){$b$}}
\put(4.3,1){\makebox(0,0){$c$}}
\put(8,0.3){\makebox(0,0){$u$}}
\put(11.3,1){\makebox(0,0){$v$}}
\put(1.8,0.5){\makebox(0,0){$G$}}
\end{picture}\caption{$G$ is a non-quasi-regularizable graph with
\textrm{core}$(G)=\{a,b,c,u\}$.}%
\label{fig24}%
\end{figure}
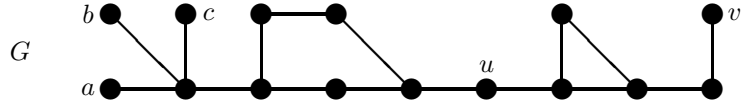

\begin{theorem}
\label{th77}For a graph $G=(V,E)$ of order $n$, the following assertions are true:

\emph{(i)} the function $d$ is supermodular, i.e., $d(A\cup B)+d(A\cap B)\geq
d(A)+d(B)$ for every $A,B\subseteq V(G)$;

\emph{(ii)} if $A$ and $B$ are critical in $G$, then $A\cup B$ and $A\cap B$
are critical as well;

\emph{(iii)} $\mathrm{\ker}(G)=\cap\left\{  B:B\text{ is a critical set of
}G\right\}  $.
\end{theorem}

\begin{proof}
\emph{(i)} Let us notice that $N(A\cup B)=N(A)\cup N(B)$ and $N(A\cap
B)\subseteq N(A)\cap N(B)$. Further, we obtain%
\begin{align*}
d(A\cup B)  &  =\left\vert A\cup B\right\vert -\left\vert N(A\cup
B)\right\vert =\left\vert A\cup B\right\vert -\left\vert N(A)\cup
N(B)\right\vert =\\
&  =\left\vert A\right\vert +\left\vert B\right\vert -\left\vert A\cap
B\right\vert -\left\vert N(A)\right\vert -\left\vert N(B)\right\vert
+\left\vert N(A)\cap N(B)\right\vert =\\
&  =\left(  \left\vert A\right\vert -\left\vert N(A)\right\vert \right)
+\left(  \left\vert B\right\vert -\left\vert N(B)\right\vert \right)
+\left\vert N(A)\cap N(B)\right\vert -\left\vert A\cap B\right\vert =\\
&  =d(A)+d(B)-\left(  \left\vert A\cap B\right\vert -\left\vert N(A\cap
B)\right\vert \right)  +\left\vert N(A)\cap N(B)\right\vert -\left\vert
N(A\cap B)\right\vert =\\
&  =d(A)+d(B)-d(A\cap B)+\left\vert N(A)\cap N(B)\right\vert -\left\vert
N(A\cap B)\right\vert \geq\\
&  \geq d(A)+d(B)-d(A\cap B).
\end{align*}

\emph{(ii)} By part \emph{(i)}, we have that
\[
d(A\cup B)+d(A\cap B)\geq d(A)+d(B)=2d_{c}(G).
\]
Consequently, we get that $d(A\cup B)=d(A\cap B)=d_{c}(G)$, i.e., both $A\cup
B$ and $A\cap B$ are critical sets.

\emph{(iii) }Let $\Gamma_{ci}$ be the family of all critical independent sets
of $G$, while $\Gamma_{c}$ denotes the family $\left\{  B:B\text{ \textit{is a
critical set in }}G\right\}  $.

By part \emph{(ii)}, both sets
\[
\mathrm{\ker}(G)=\cap\left\{  S:S\in\Gamma_{ic}\right\}  \text{ and }%
Q_{c}=\cap\left\{  B:B\in\Gamma_{c}\right\}
\]
are critical. Theorem \ref{Theorem3} implies that $\Gamma_{ci}\subseteq
\Gamma_{c}$, and therefore, $Q_{c}\subseteq\mathrm{\ker}(G)$. On the other
hand, $Q_{c}$ is independent, because by Theorem \ref{Theorem3}, one of the
critical sets from $\Gamma_{c}$ is independent. Consequently, we obtain
$\mathrm{\ker}(G)\subseteq Q_{c}$, and this completes the proof.
\end{proof}

\begin{theorem}
\label{th7}For a graph $G=(V,E)$ of order $n$, the following assertions are true:

\emph{(i)} $V\supseteq\mathrm{corona}(G)\supseteq S\supseteq\mathrm{core}%
(G)\supseteq\mathrm{\ker}(G)$, for every $S\in\Omega(G)$;

\emph{(ii)} $n\geq\zeta\left(  G\right)  \geq\alpha\left(  G\right)  \geq
\xi\left(  G\right)  \geq\varepsilon\left(  G\right)  \geq d_{c}\left(
G\right)  \geq\alpha\left(  G\right)  -\mu\left(  G\right)  $;

\emph{(iii) }$\xi\left(  G\right)  \geq\alpha\left(  G\right)  -\mu\left(
G\right)  +\varepsilon\left(  G\right)  -d_{c}\left(  G\right)  $.
\end{theorem}

\begin{proof}
\emph{(i)} Clearly, $\mathrm{core}(G)\subseteq S\subseteq\mathrm{corona}%
(G)\subseteq V$ hold for each $S\in\Omega\left(  G\right)  $. The set
$\mathrm{\ker}(G)$ is independent by definition. According to Theorem
\ref{th77}\emph{(ii)}, $\mathrm{\ker}(G)$ is critical. Consequently, by
Theorem \ref{th8}\emph{(iv)}, there exists a matching $M_{L}$ from
$N(\mathrm{\ker}(G))$ into $\mathrm{\ker}(G)$. Figure \ref{fig123} will
accompany us all the way to the end of the proof.\begin{figure}[h]
\setlength{\unitlength}{1.0cm} \begin{picture}(5,3.5)\thicklines
\put(6.7,1){\oval(11.7,1.5)}
\put(4.5,1){\oval(5,1)}
\put(4.5,3){\oval(5,1)}
\multiput(2.5,1)(1,0){2}{\circle*{0.29}}
\multiput(2.5,3)(1,0){2}{\circle*{0.29}}
\multiput(5.5,1)(0,2){2}{\circle*{0.29}}
\multiput(3.6,1)(0.35,0){5}{\circle*{0.1}}
\multiput(3.6,3)(0.35,0){5}{\circle*{0.1}}
\multiput(6.4,3)(0.3,0){5}{\line(1,0){0.15}}
\multiput(5.5,3)(0.5,-0.5){4}{\line(1,-1){0.4}}
\put(7.5,1){\circle*{0.29}}
\put(8,3){\circle*{0.29}}
\put(8,2.6){\makebox(0,0){$u$}}
\put(7.5,0.6){\makebox(0,0){$v$}}
\put(5.5,1){\line(0,1){2}}
\multiput(2.5,1)(1,0){2}{\line(0,1){2}}
\put(6.5,1){\line(0,1){2}}
\multiput(6.5,1)(0,2){2}{\circle*{0.29}}
\put(0.5,1){\makebox(0,0){$S$}}
\put(0.2,2){\makebox(0,0){$G$}}
\put(1.5,1){\makebox(0,0){$A_{2}$}}
\put(1,3){\makebox(0,0){$ker(G)-A_{1}$}}
\put(10,1){\oval(3.3,1)}
\put(10.5,3){\oval(2.3,1)}
\multiput(9,1)(1,0){3}{\circle*{0.29}}
\multiput(10,3)(1,0){2}{\circle*{0.29}}
\put(9,1){\line(1,2){1}}
\put(10,1){\line(1,2){1}}
\put(12,1){\makebox(0,0){$A_{1}$}}
\put(12.3,3){\makebox(0,0){$N(A_{1})$}}
\put(2.1,2.3){\makebox(0,0){$M_{B}$}}
\put(9.2,2.3){\makebox(0,0){$M_{L}$}}
\end{picture}\caption{$S\in\Omega(G)$, $\mathrm{\ker}(G)$, and $A_{1}%
=S\cap\mathrm{\ker}(G)${.}}%
\label{fig123}%
\end{figure}
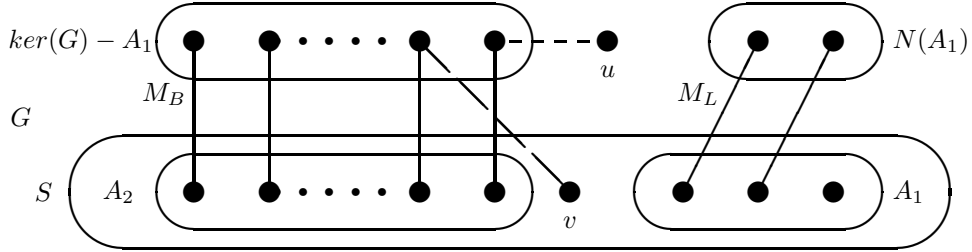

Let $S\in\Omega(G)$, and $A_{1}=\mathrm{\ker}(G)\cap S$. Since $\mathrm{\ker
}(G)-A_{1}$ is stable and disjoint from $S$, Theorem \ref{th2} ensures that
there is a matching $M_{B}$ from $\mathrm{\ker}(G)-A_{1}$ into $S$, covering
some subset $A_{2}$ of $S-A_{1}$ Let $S\in\Omega(G)$, and $A_{1}=\mathrm{\ker
}(G)\cap S$. Since $\mathrm{\ker}(G)-A_{1}$ is stable and disjoint from $S$,
Theorem \ref{th2} ensures that there is a matching $M_{B}$ from $\mathrm{\ker
}(G)-A_{1}$ into $S$, covering some subset $A_{2}$ of $S-A_{1}$. Clearly, we
have
\[
\left\vert \mathrm{\ker}(G)-A_{1}\right\vert =\left\vert A_{2}\right\vert
\text{, }A_{1}\cap A_{2}=\emptyset\text{, and }A_{2}\subseteq N(\mathrm{\ker
}(G)-A_{1})\cap S.
\]

Assume that there is some $v\in\left(  N(\mathrm{\ker}(G)-A_{1})\cap S\right)
-A_{2}$. The vertex $v$ must be matched with some vertex from $\mathrm{\ker
}(G)-A_{1}$ by $M_{L}$, because $\{v\}\cup A_{1}\subseteq S$. Hence $M_{L}$
matches the set $N(\mathrm{\ker}(G)-A_{1})\cap S$ into $\mathrm{\ker}%
(G)-A_{1}$, which is impossible, since
\[
\left\vert N(\mathrm{\ker}(G)-A_{1})\cap S\right\vert \geq\left\vert \{v\}\cup
A_{2}\right\vert >\left\vert A_{2}\right\vert =\left\vert \mathrm{\ker
}(G)-A_{1}\right\vert .
\]
Consequently, we get that $N(\mathrm{\ker}(G)-A_{1})\cap S=A_{2}$. Thus
$M_{L}$ matches the set $N(\mathrm{\ker}(G)-A_{1})\cap S$ onto $\mathrm{\ker
}(G)-A_{1}$, and $N(A_{1})$ into $A_{1}$. Clearly, we have
\[
\left\vert \mathrm{\ker}(G)-A_{1}\right\vert =\left\vert A_{2}\right\vert
\text{, }A_{1}\cap A_{2}=\emptyset\text{, and }A_{2}\subseteq N(\mathrm{\ker
}(G)-A_{1})\cap S.
\]

Assume that there is some $v\in\left(  N(\mathrm{\ker}(G)-A_{1})\cap S\right)
-A_{2}$. The vertex $v$ must be matched with some vertex from $\mathrm{\ker
}(G)-A_{1}$ by $M_{L}$, because $\{v\}\cup A_{1}\subseteq S$. Hence $M_{L}$
matches the set $N(\mathrm{\ker}(G)-A_{1})\cap S$ into $\mathrm{\ker}%
(G)-A_{1}$, which is impossible, since
\[
\left\vert N(\mathrm{\ker}(G)-A_{1})\cap S\right\vert \geq\left\vert \{v\}\cup
A_{2}\right\vert >\left\vert A_{2}\right\vert =\left\vert \mathrm{\ker
}(G)-A_{1}\right\vert .
\]
Consequently, we get that $N(\mathrm{\ker}(G)-A_{1})\cap S=A_{2}$. Thus
$M_{L}$ matches the set $N(\mathrm{\ker}(G)-A_{1})\cap S$ onto $\mathrm{\ker
}(G)-A_{1}$, and $N(A_{1})$ into $A_{1}$.

In conclusion, we may assert that $\left\vert \mathrm{\ker}(G)\right\vert
-\left\vert N(\mathrm{\ker}(G))\right\vert =\left\vert A_{1}\right\vert
-\left\vert N(A_{1})\right\vert $. Hence, we infer that $\mathrm{\ker
}(G)-A_{1}=\emptyset$, otherwise we have that $A_{1}$ is a critical
independent set of $G$ with $\left\vert A_{1}\right\vert <\left\vert
\mathrm{\ker}(G)\right\vert $, in contradiction with the hypothesis on
minimality of $\mathrm{\ker}(G)$. This ensures that $\mathrm{\ker}(G)\subseteq
S$ for every $S\in\Omega(G)$, which means that $\mathrm{\ker}(G)\subseteq
\mathrm{core}(G)$.

\emph{(ii)} Using part \emph{(i)}, Theorem \ref{th77}\emph{(iii)}, and
Corollary \ref{cor1}, we deduce that
\begin{gather*}
n\geq\zeta\left(  G\right)  \geq\alpha\left(  G\right)  \geq\xi\left(
G\right)  \geq\varepsilon(G)=\\
=\left\vert \mathrm{\ker}(G)\right\vert \geq\left\vert \mathrm{\ker
}(G)\right\vert -\left\vert N\left(  \mathrm{\ker}(G)\right)  \right\vert
=d_{c}(G)\geq\alpha(G)-\mu(G),
\end{gather*}
which completes the proof.

\emph{(iii) }It follows immediately from part \emph{(ii)}.
\end{proof}

Notice that $\xi(K_{2,3})=\varepsilon(K_{2,3})>d_{c}(K_{2,3})=1=\alpha
(K_{2,3})-\mu(K_{2,3})$, while the graph $G_{2}$ is from Figure \ref{fig3}
satisfies $\xi(G_{2})>\varepsilon(G_{2})>d(G_{2})=1$.

\begin{corollary}
If $d_{c}\left(  G\right)  >0$ or, equivalently, $G$ is a
non-quasi-regularizable graph,$\ $then

\emph{(i) }$n\geq\zeta\left(  G\right)  \geq\alpha\left(  G\right)  \geq
\xi\left(  G\right)  \geq\varepsilon\left(  G\right)  >d_{c}\left(  G\right)
\geq\alpha\left(  G\right)  -\mu\left(  G\right)  \geq1$;

\emph{(ii) }$\xi\left(  G\right)  >\alpha\left(  G\right)  -\mu\left(
G\right)  +\varepsilon\left(  G\right)  -d_{c}\left(  G\right)  .$
\end{corollary}

\begin{proof}
According to Theorem \ref{th11}, $G$ is non-quasi-regularizable if and only if
$\mathrm{\ker}(G)\neq\emptyset$, i.e., $\left\vert \mathrm{\ker}(G)\right\vert
\geq2$. The fact that $G$ has no isolated vertices implies $N\left(
\mathrm{\ker}(G)\right)  \neq\emptyset$, and consequently, it follows
$\varepsilon(G)=\left\vert \mathrm{\ker}(G)\right\vert >\left\vert
\mathrm{\ker}(G)\right\vert -\left\vert N\left(  \mathrm{\ker}(G)\right)
\right\vert =d_{c}(G)$. Further, using Theorem \ref{th7}, we get both
\emph{(i)} and \emph{(ii)}.
\end{proof}

\begin{corollary}
\cite{ChCh2008} If there is some $S\in\mathrm{Ind}(G)$ with $\left\vert
S\right\vert >\left\vert N(S)\right\vert $, then $\xi\left(  G\right)
>d_{c}\left(  G\right)  $.
\end{corollary}

\section{Conclusions}

Writing this paper we have been motivated by the inequality
\[
\xi(G)=\left\vert \mathrm{core}(G)\right\vert >\alpha(G)-\mu(G),
\]
which is true for every graph $G$ without isolated vertices, such that
$\alpha(G)>\mu(G)$ \cite{BorosGolLev}. What we have found is that there exists
a subset of $\mathrm{core}(G)$, which is a real obstacle to its nonemptiness.
The cardinality of this subset, namely, $\varepsilon\left(  G\right)
=\left\vert \mathrm{\ker}(G)\right\vert $ stands out above $\alpha(G)-\mu(G)$
on its own.

The problem of whether there are vertices in a given graph $G$ belonging to
$\mathrm{core}(G)$ is \textbf{NP}-hard \cite{BorosGolLev}. On the other hand,
it has been noticed that for some families of graphs $\mathrm{core}\left(
G\right)  $ may be computed in polynomial time.

We conclude with the following question.

\begin{problem}
Is it true that for any fixed positive integer $k$, to decide if
$\varepsilon\left(  G\right)  >k$ is \textbf{NP}-complete?
\end{problem}


\begin{thebibliography}{99}                                                                                               %


\bibitem {Ageev}A. A. Ageev, \emph{On finding critical independent and vertex
sets}, SIAM Journal of discrete mathematics \textbf{7} (1994) 293-295.

\bibitem {berge2}C. Berge, \emph{Some common properties for regularizable
graphs, edge-critical graphs and }$B$\emph{-graphs}, In: Rosa, A., Sabidussi,
G., Turgeon, J., eds., Theory and Practice of Combinatorics. North-Holland
Mathematics Studies \textbf{60}, Amsterdam: North--Holland, pp. 31-44.

\bibitem {BorosGolLev}E. Boros, M. C. Golumbic, V. E. Levit, \emph{On the
number of vertices belonging to all maximum stable sets of a graph}, Discrete
Applied Mathematics \textbf{124} (2002) 17--25.

\bibitem {Butenko}S. Butenko, S. Trukhanov, \emph{Using critical sets to solve
the maximum independent set problem}, Operations Research Letters \textbf{35}
(2007) 519-524.

\bibitem {ChCh2008}M. Chleb\'{\i}k, J. Chleb\'{\i}kov\'{a}, \emph{Crown
reductions for the minimum weighted vertex cover problem}, Discrete Applied
Mathematics \textbf{156} (2008) 292-312.

\bibitem {dem}R. W. Deming, \emph{Independence numbers of graphs - an
extension of the K\"{o}nig-Egerv\'{a}ry theorem}, Discrete Mathematics
\textbf{27} (1979) 23-33.

\bibitem {GunHarRall}G. Gunther, B. Hartnell, D. F. Rall, \emph{Graphs whose
vertex independence number is unaffected by single edge addition or deletion},
Discrete Applied Mathematics \textbf{46} (1993) 167-172.

\bibitem {HamHanSim}P. L. Hammer, P. Hansen, B. Simeone, \emph{Vertices
belonging to all or to no maximum stable sets of a graph}, SIAM Journal of
Algebraic Discrete Methods \textbf{3} (1982) 511-522.

\bibitem {Jamison}R. E. Jamison, \emph{Alternating Whitney sums and matchings
in trees, part 1}, Discrete Mathematics \textbf{67} (1987) 177-189.

\bibitem {Larson}C. E. Larson, \emph{A note on critical independence
reductions}, Bulletin of the Institute of Combinatorics and its Applications
\textbf{5} (2007) 34-46.

\bibitem {Larson2011}C. E. Larson, \emph{The critical independence number and
an independence decomposition}, European Journal of Combinatorics \textbf{32}
(2011) 294-300.

\bibitem {LevMan2001}V. E. Levit, E. Mandrescu, \emph{On the structure of
}$\alpha$\emph{-stable graphs}, Discrete Mathematics \textbf{236} (2001) 227-243.

\bibitem {LevMan2002a}V. E. Levit, E. Mandrescu, \emph{Combinatorial
properties of the family of maximum stable sets of a graph}, Discrete Applied
Mathematics \textbf{117} (2002) 149-161.

\bibitem {LevMan2002b}V. E. Levit, E. Mandrescu, \emph{A new Greedoid: the
family of local maximum stable sets of a forest}, Discrete Applied Mathematics
\textbf{124} (2002) 91-101.

\bibitem {LevMan2003}V. E. Levit, E. Mandrescu, \emph{On }$\alpha^{+}%
$\emph{-stable K\"{o}nig-Egerv\'{a}ry graphs}, Discrete Mathematics
\textbf{263} (2003) 179-190.

\bibitem {LevManKE2009}V. E. Levit, E. Mandrescu, \emph{Critical independent
sets and K\"{o}nig-Egerv\'{a}ry graphs}, Graphs and Combinatorics (2011)
(accepted), math.CO.arXiv:0906.4609, 8 pp.

\bibitem {Lorentzen1966}L. C. Lorentzen. \emph{Notes on covering of arcs by
nodes in an undirected graph}, Technical report ORC 66-16, Operations Research
Center, University of California, Berkley, California, 1966.

\bibitem {NemhTro}G. L. Nemhauser, L. E. Trotter, Jr., \emph{Vertex packings:
structural properties and algorithms}, Mathematical Programming \textbf{8}
(1975) 232-248.

\bibitem {Schrijver2003}A. Schrijver, \emph{Combinatorial Optimization},
Springer, Berlin, 2003.

\bibitem {ster}F. Sterboul, \emph{A characterization of the graphs in which
the transversal number equals the matching number}, Journal of Combinatorial
Theory Series B \textbf{27} (1979) 228-229.

\bibitem {Tutte}W. T. Tutte, \emph{The }$1$\emph{-factors of oriented graphs},
Proceedings of the American Mathematical Society \textbf{22} (1947) 107-111.

\bibitem {Zhang}C. Q. Zhang, \emph{Finding critical independent sets and
critical vertex subsets are polynomial problems}, SIAM Journal of Discrete
Mathematics \textbf{3} (1990) 431-438.

\bibitem {Zito}J. Zito, \emph{The structure and maximum number of maximum
independent sets in trees}, Journal of Graph Theory \textbf{15} (1991) 207-221.
\end{thebibliography}
\end{document}